\documentclass[pdf-latex]{sn-jnl}

\usepackage{adjustbox}

\usepackage{amsmath}
\usepackage{amsthm} 
\usepackage{amssymb}
\newtheorem{theorem}{Theorem}

\newtheorem{prop}{Proposition}[theorem]

\usepackage[style=numeric,backend=biber]{biblatex}
\begin{document}

\title[Article Title]{Can solvents tie knots? Helical folds of biopolymers in liquid environments.}

\author*[1]{\fnm{Rhoslyn} \sur{Coles}}\email{rhoslyn.coles@mathematik.tu-chemnitz.de}

\author*[2]{\fnm{Myfanwy E.} \sur{Evans}}\email{evans@uni-potsdam.de}

\affil[1]{\orgdiv{Fakult\"at f\"ur Mathematik}, \orgname{Technische Universit\"at Chemnitz}, \orgaddress{\city{Chemnitz}, \postcode{09107}, \country{Germany}}}

\affil[2]{\orgname{University of Potsdam}, \orgdiv{Institute for Mathematics},  \city{Potsdam}, \postcode{14476}, \country{Germany}}

\abstract{
Helices are the quintessential geometric motif of the microscale, from {$\alpha$--helices} in proteins to double helices in DNA. 
Assembly of the helical biopolymers is a foundational step in a hierarchy of structure that leads to biological activity.  
By simulating folding in a simplified setting we probe the role of the solvent in the collaborative processes governing biomaterials. 
Using a simulation technique based on the morphometric approach to solvation, we performed computer experiments in which a short, flexible tube---modelling a biopolymer in an aqueous environment---folds solely based on the interaction of the tube with the solvent.  
Our findings reveal a variety of helical structures that assemble depending on solvent conditions, including overhand knots and symmetric double helices.
By differentiating the role of solvation, our work illuminates the environment of all soluble biomolecules, demonstrating that the solvent can drive fundamental rearrangements, even up to tying a simple overhand knot.}


\maketitle

\section{Introduction}\label{sec1}

Solvation, the physics of molecules in fluids, is the collection of spontaneous processes resulting in the energetically favourable (re)arrangement of the molecule within the fluid. 
Solvation is one of the key mechanisms through which the aqueous environment of soluble biomolecules, such as proteins and nucleic acids, affects their structure, stability and functioning~\cite{LEV06, BAL08, PRI17}.
At a fundamental level, the solvent mediates between chemistry and structure, as a chaperon to other more concrete processes. 
The capacity of the solvent \emph{by itself} receives little attention; the water surrounding is arguably the most challenging part of the modelling of a soluble biomolecule and therefore the development of accurate yet computationally efficient models is of central importance~\cite{ONU17}.
Details tend to be on the atomic level, specific to a particular protein, and do not address the behaviour of a general soluble polymer in a general solvent.

Biopolymer chains wind themselves into a plethora of shapes in nature. 
The optimal $\alpha$-helix and $\beta$-sheet forms are the fundamental motifs found in protein structures. 
DNA, itself a double helix, can exhibit a variety of geometric forms, which helps to expose certain base pairs along the strand \cite{IRO15}.
Knotted configurations are found in biopolymers in wide variety of settings, where their geometric form is related to their functionality \cite{MAR09,SHA11}. 
The broader zoo of optimal shapes of short flexible biopolymers (see Fig.~\ref{fig:intro}), and their contribution to biological function, is a rich field of research. 
From this broader perspective, we investigate how a fluid environment may affect the shape of a short tube--like string. 
Understanding form through experimentation with simple geometric objects under physically motivated constraints provides interesting insights; our optimal forms contain helical motifs known for their optimal packing upon confinement~\cite{MAR00, POL08}, and our results establish the thermodynamic stability of the overhand knot and double helix in solution.

Modelling the effect of the solvent on the structure of large (polyatomic) biomolecules like globular proteins is challenging: A protein's configuration is governed by a fine balance between intramolecular bonding energy and the free energy of solvation of the fluid system~\cite{DOB03, SAL03}.
Whilst including the solvent explicitly in state models may enable a more complete picture of the solvent--solute interaction, the additional computational cost means there is little possibility to explore complex solute geometries let alone solvent induced shape change~\cite{LEV06}.
These difficulties motivate the development of implicit solvent models, which treat the solvent as a continuous medium, with application in molecular dynamics simulations to efficiently simulate biomolecules like proteins and nucleic acids in solution \cite{ONU19, ROU99}.
Implicit solvent models make use of the geometric properties of the space occupied by the solute within the liquid in order to compute the free energy of the liquid~\cite{EDE05}.
Of interest therefore is the space-filling representation of a molecule or solute expanded by the solvent's radius, known as the \textit{solvent-accessible surface} \cite{LEE71}. 
Early solvation free energy models were often based on the volume excluded by this surface which has an entropic cost to the energy~\cite{SIM94,LUM99}, as well as the surface area~\cite{BAN07, OOI87,EIS86}.

\begin{figure}[h!]
\centering
\includegraphics[width=0.65\linewidth]{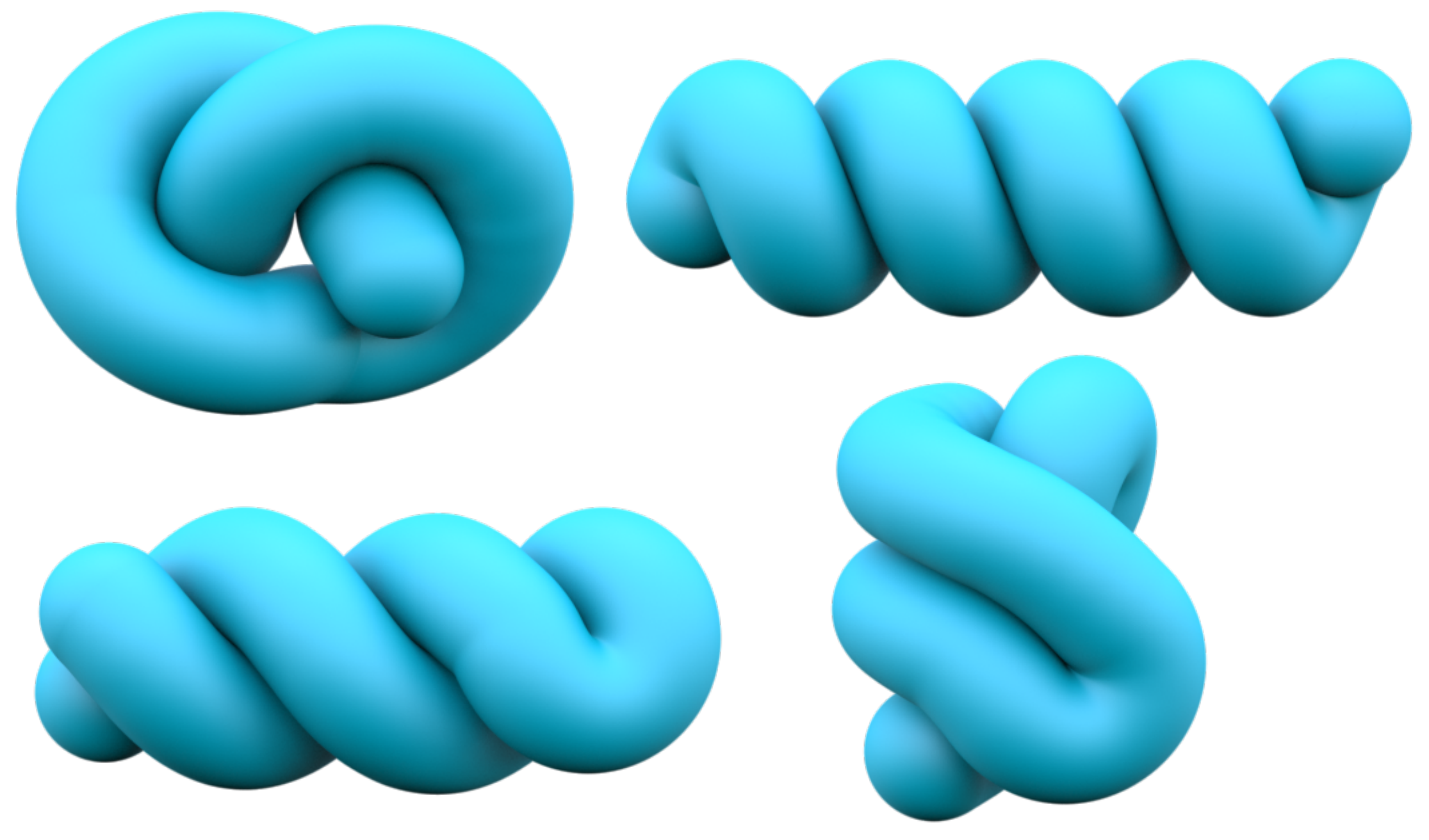}
\caption{Helical forms of a flexible tube, relevant to biopolymer foldings. Clockwise from top left, an overhand knot, the optimal helix, a compact structure with curve motifs arranged in parallel, and a double helix folded back on itself.}
\label{fig:intro}
\end{figure}

This link between geometry and the thermodynamics of fluids was made more precise through the development of the morphometric approach to modelling solvation~\cite{MEC96, ROT06, KOE04}. In this approach the solvation free energy, $\mathbf{F}_{\mathrm{sol}}$ in~\eqref{eq:free_energy_solvation}, is given as a linear sum of the basic rigid motion invariant valuations, the volume $\mathrm{V}$, surface area $\mathrm{A}$ and two further measures of curvature $\mathrm{C}$ and $\mathrm{X}$, of the body bounded by the solvent accessible surface ($\mathcal{B}$).
The thermodynamic coefficients coupling to the geometric measures are $\mathrm{p}>0$ the fluid pressure, $\sigma<0$ the planar surface tension and $\kappa$, $\bar{\kappa} \in \mathbb{R}$ for which there is no physical interpretation~\cite{KOE04, ROT06, HAD57}.
\begin{equation}\label{eq:free_energy_solvation}
\mathbf{F}_{\mathrm{sol}} = \mathrm{p} \mathrm{V}(\mathcal{B}) + \sigma  \mathrm{A}(\mathcal{B}) + \kappa \mathrm{C}(\mathcal{B}) + \bar{\kappa}\mathrm{X}(\mathcal{B}) \,.
\end{equation}
These four geometric functions, the \emph{measures of curvature}, appear as a novel application of Hadwiger's characterisation theorem of integral geometry~\cite{HAD57}. 
The measures $\mathrm{C}$ and $\mathrm{X}$ are computed as the mean and Gaussian curvatures integrated over the boundary of $\mathcal{B}$ provided sufficient regularity allows an interpretation of these curvature functions~\cite{ZAE86}. 
By the Gauss--Bonnet theorem $\mathrm{X} = 4\pi\chi$ where $\chi$ is the Euler characteristic.

The real advantage of this approach is that it (de)couples physics and geometry in a computationally convenient manner.
The thermodynamic coefficients $\mathrm{p}$, $\sigma$, $\kappa$ and $\bar{\kappa}$ depend only the physical properties of the fluid, like temperature and the chemical coupling between the solvent and solute, \emph{not} on the shape of the solute. 
These can be determined by fitting free energy values as obtained from~\eqref{eq:free_energy_solvation} to those computed from solvent models of the statistical--mechanical theory, as tested with simple solute geometries. Once the coefficients are given, the free energy is evaluated from the geometry of $\mathcal{B}$. 
If $\mathcal{B}$ is modelled as a union of balls, results from computational topology lead to the development of fast algorithms which evaluate $\mathrm{V}(\mathcal{B})$, $\mathrm{A}(\mathcal{B})$, $\mathrm{C}(\mathcal{B})$ and $\mathrm{X}(\mathcal{B})$ exactly and efficiently, fairly independently of the shape~\cite{EDE95, EDE05}.
For fluid systems like that of a protein within the aqueous environment of the cell, the morphometric approach can compute free energy values of complex solute geometries in excellent agreement with the classical theory at a fraction of the computational cost, making a geometric centered investigation of solvation even feasible~\cite{ROT06}.

Previously the morphometic approach was used to demonstrate that different solvent environments favoured different helical tubular solutes~\cite{HAN07}. 
This was shown by comparing the solvation free energy between a collection of \emph{tight} periodic helical tubes, winding such that each successive turn sits on the previous~\cite{PRZ01}.
Energetically favourable configurations included the $\alpha$--helix, slightly unwound helices, resembling topologically an open infinite cylinder, and stacked parallel curves representing the infinite $\beta$--sheet structure.
Our work is a considerable extension of this study by allowing the tubular solute to freely fold \emph{without a priori} assuming helical like curve arrangements. 
When considering finite strings the preference for tight single helical structures is challenged by our findings here.
The morphometric approach has also been used in a variety of settings to analyse minimising configurations~\cite{ROT06, POL08, HAR13,EVA14}, as well as for the geometric simulation of hard sphere clusters in fluids, where helical stacks of sphere were found under particular fluid conditions \cite{SPI24}. 

In this study, we will use the morphometric approach to solvation as a basis for simulating the folding of a finite flexible tube, modelling a biopolymer in aqueous environments, according to the interaction of the tube with the solvent alone. We detail the results of these simulations below.

\subsection{Helical folds in solvation simulations}\label{sec:results}

We simulated the folding of energetically favourable configurations of a (unit radius) tube of length $\ell = 25$ within a range of fluid conditions.
This was achieved through computer experiments which optimise the shape of an open equilateral polygonal curve of fixed length, modelling the shape of the solute, according to the morphometric description of the free energy of solvation~\eqref{eq:free_energy_solvation}.
The free energy is minimised via the method of simulated annealing: the curve shape is incrementally improved using a crank-shaft move while ensuring that the flexible tube does not intersect itself (see methods for details).
The assembly of a double helical configuration and overhand knot via this computation method is shown in Fig.~\ref{fig:deformation_sequence} and Fig.~\ref{fig:deformation_sequence_tied}.
The simulation is initialised in the fully solvated state and progresses to fold. 
First the tube comes into self--contact, causing the solvent--accessible surface to self--intersect thereby decreasing the volume and, as the leading order term in~\eqref{eq:free_energy_solvation}, thus the energy.
The shape then advances as controlled by the specific linear combination of the measures given in~\eqref{eq:free_energy_solvation} i.e. by the fluid environment.
The deformation shown in Fig.~\ref{fig:deformation_sequence} appears to be the most direct and therefore fastest mechanism with which the curve adopts configuration $\mathbf{B}$.
The configuration may initially fold differently (if for example the kinked structure shown in Fig.~\ref{fig:deformation_sequence} A is not approximately folding the curve in half) leading to a double helical curve with one length shorter than the other. 
Such versions are then iteratively improved via successive rounds of heating and cooling.
On the other hand, whilst the deformation shown in Fig.~\ref{fig:deformation_sequence_tied} is typical, there are many ways in which the curve may fold into shape $\mathbf{C}$.
The number of iterations steps needed to fold the structure depends upon the fluid conditions and the folding shape, as well as the specific implementation of the algorithm i.e. the expected distance points are moved between iterations and the tempering scheme employed.
An approximate upper bound of $10^{6}$ steps is required to fold the curve into a low energy configuration.
\begin{figure}[htbp!]
\centering
\adjustbox{valign=t}{\begin{minipage}{0.6\textwidth} 
    \includegraphics[width=\linewidth]{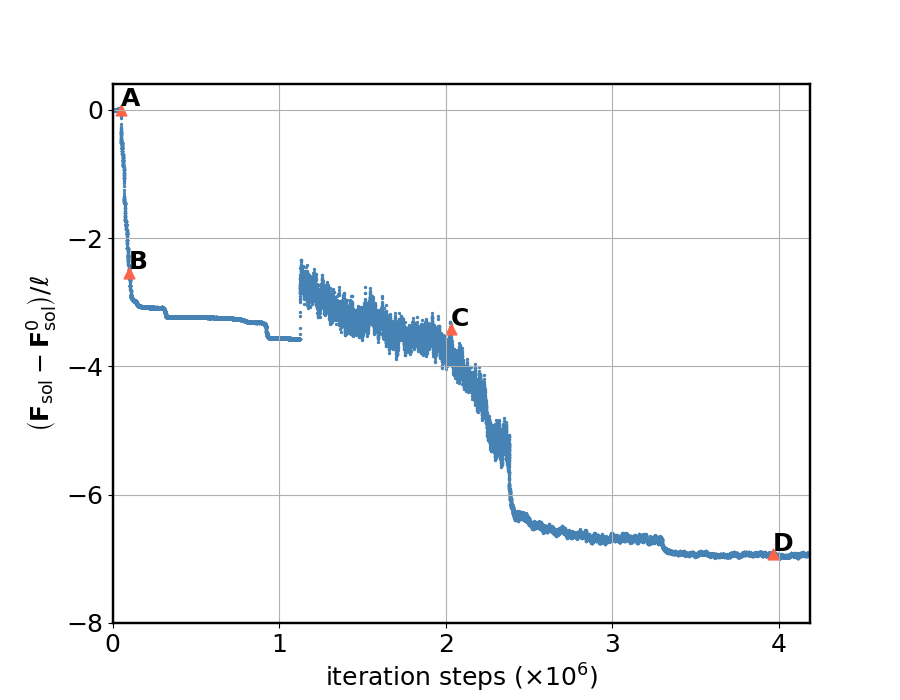}
\end{minipage}}
\hfill
\adjustbox{valign=t}{\begin{minipage}{0.38\textwidth}
    \vspace{7.5mm}
    \includegraphics[width=\linewidth]{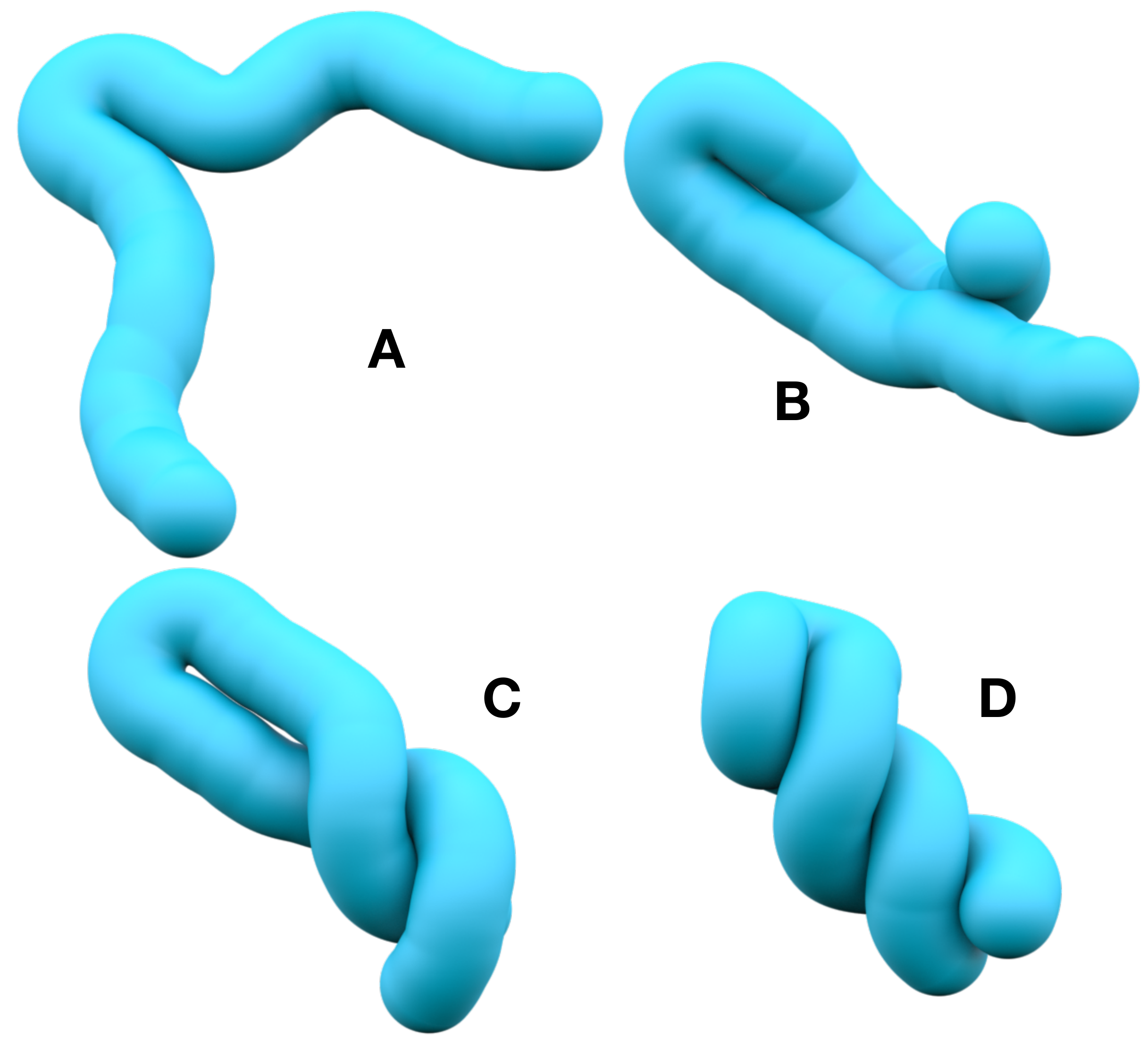}
    \vfill
\end{minipage}}
\caption{Simulated folding of double helical geometry. The solvation free energy of a tube during a simulation run, in units $(\mathbf{F}_{\mathrm{sol}}-\mathbf{F}_{\mathrm{sol}}^{0})/\ell$ (energy per unit length relative to the \emph{fully solvated} state). In configurations \textbf{A} and \textbf{B} the tube comes into self--contact causing the solvent--accessible surface (not shown) to self--intersect thereby decreasing the volume and, as the leading order term in~\eqref{eq:free_energy_solvation}, thereby the energy. 
In configurations \textbf{C} and \textbf{D} the shape assembles as controlled by the specific linear combination of the geometric measures. 
The packing fraction is $\eta = 0.375$ and the solvent radius 
}
\label{fig:deformation_sequence}
\end{figure}

\begin{figure}[htbp!]
\centering
\adjustbox{valign=t}{\begin{minipage}{0.6\textwidth} 
    \includegraphics[width=\linewidth]{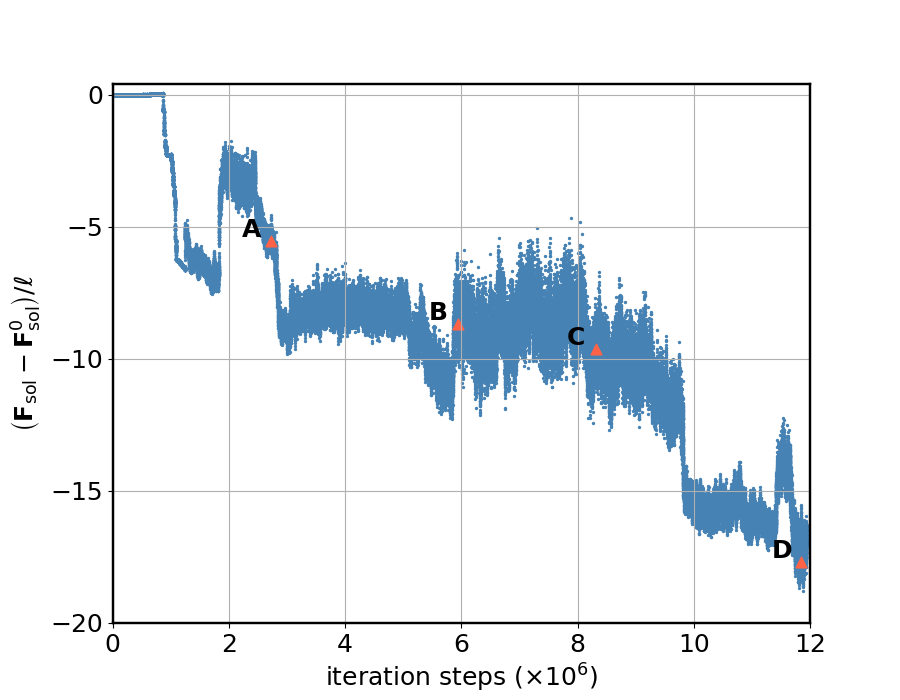}
\end{minipage}}
\hfill
\adjustbox{valign=t}{\begin{minipage}{0.38\textwidth}
    \vspace{10mm}
    \includegraphics[width=0.9\linewidth]{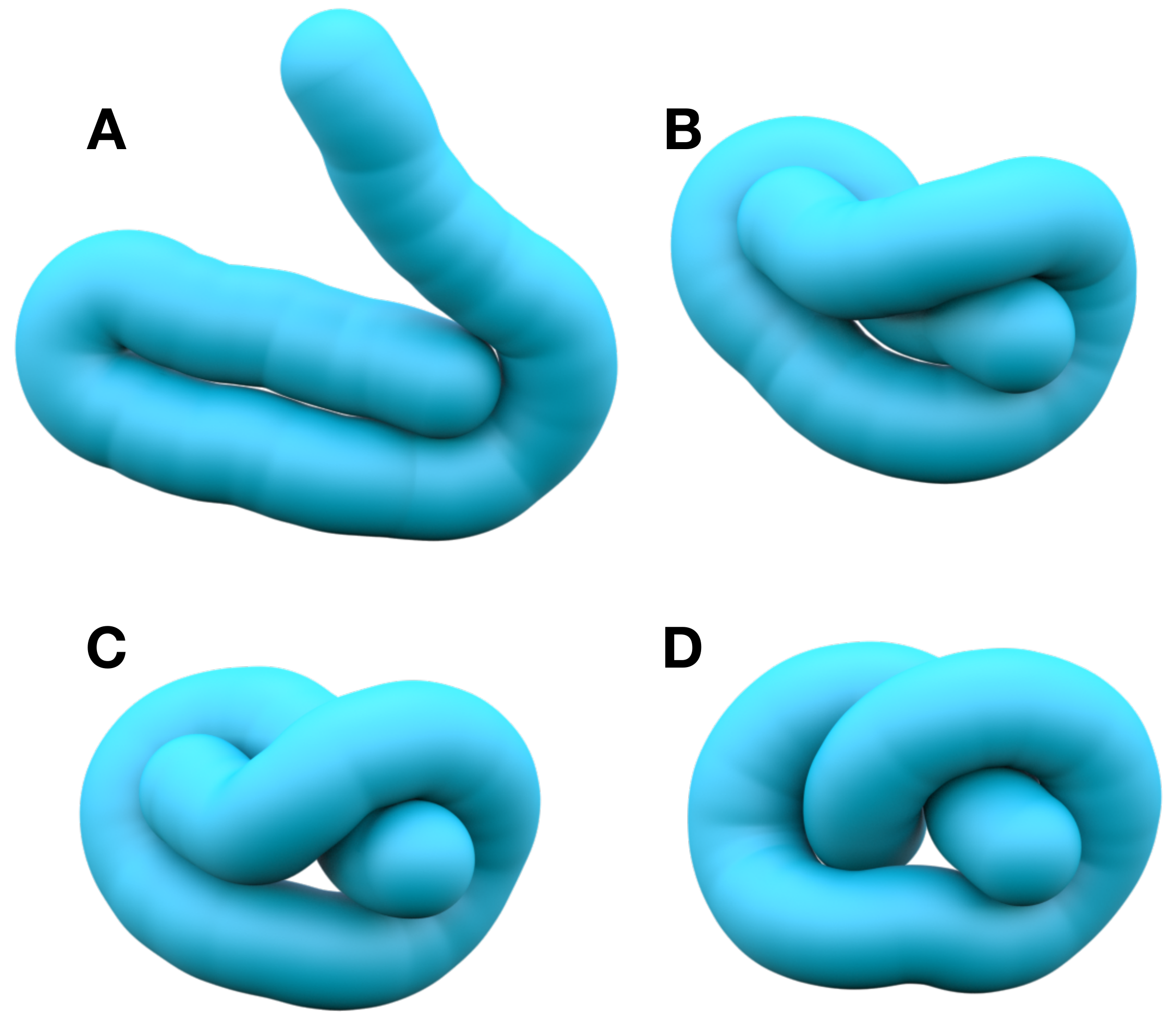}
    \vfill
\end{minipage}}
\caption{Simulated folding of overhand knot geometry. The solvation free energy of a tube during a simulation run, in units $(\mathbf{F}_{\mathrm{sol}}-\mathbf{F}_{\mathrm{sol}}^{0})/\ell$ (energy per unit length relative to the \emph{fully solvated} state). In configurations \textbf{A} and \textbf{B} the tube comes into self--contact causing the solvent--accessible surface (not shown) to self--intersect thereby decreasing the volume and, as the leading order term in~\eqref{eq:free_energy_solvation}, thereby the energy. 
In configurations \textbf{C} and \textbf{D} the shape assembles as controlled by the specific linear combination of the geometric measures. 
The packing fraction is $\eta = 0.35$ and the solvent radius 
}
\label{fig:deformation_sequence_tied}
\end{figure}

Modelling the solvent as a hard sphere fluid, the physical coefficients in~\eqref{eq:free_energy_solvation} may be derived explicitly as functional expressions of the solvent packing fraction $\eta$ and solvent radius $r_{s}$, providing a way to systematically compare favorable solute geometries in different hard sphere fluid environments~\cite{HAN06b}.
The range of fluid environments, $0.02 \leq r_{s} <0.2$ and $\eta \in (0, 0.494)$ serves as a reference model for the fluid environment of a cell under physiological conditions~\cite{ROT06}. 
The landscape of optimal configurations corresponding to such fluids are shown in the phase diagram in Fig.~\ref{fig:openChain_dl25_25_phaseDiagram}. 
A square in the diagram corresponds to coordinates $(\eta, r_{s})$ defining a particular fluid environment, the colouring gives information of the configuration of lowest energy. 
The diagram was constructed using empirical classification, where structures are grouped by key distinguishing features, despite their geometries differing slightly. 
For example, a double helix along a straight axis is grouped together with one that is slightly bent, as their distinguishing feature is their double helical character. 
Additional simulations using an input of various favourable configurations under differing fluid conditions were performed to interpolate between initial simulations in the phase diagram. 
The final results of the phase diagram were enhanced by calculations of $\mathbf{F}_{\mathrm{sol}}$ for the collection of tight helical tubes known for their high degree of thermal stability~\cite{HAN07}.

\begin{figure}[h!]
\centering
\includegraphics[width=\linewidth]{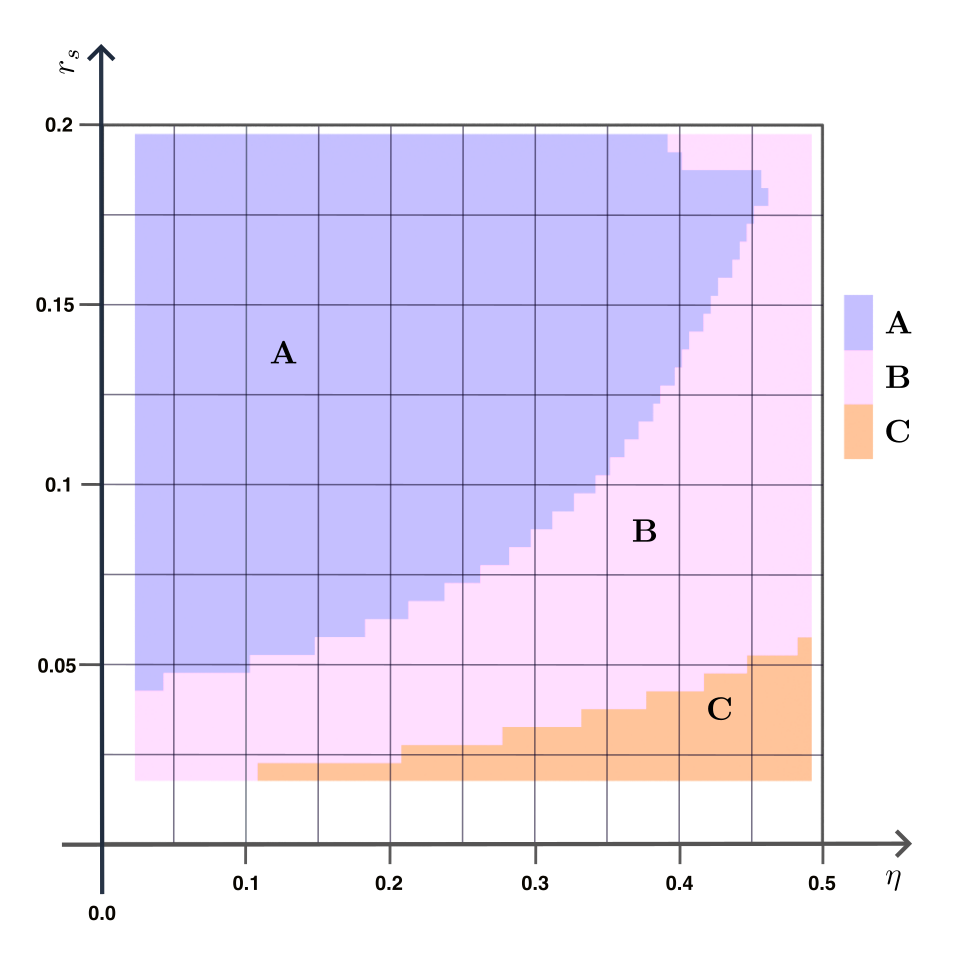}
\includegraphics[width=0.9\linewidth]{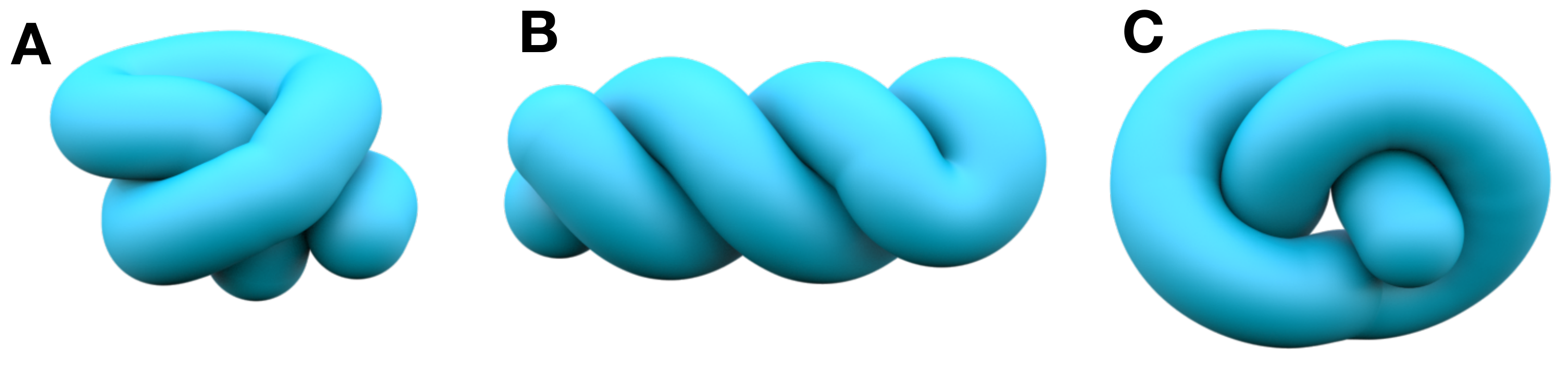}
\caption{Phase diagram of optimal string geometries over the space of fluid properties. 
The fluid is defined by the solvent packing fraction $\eta$ and the relative size of the solvent given by the solvent radius $r_s$~\cite{HAN06b}.
The diagram splits into three regions, each with a different minimising geometry. 
These three geometries are shown below the diagram, and are called configurations \textbf{A}, \textbf{B} (a double helix), and \textbf{C} (an overhand knot).}
\label{fig:openChain_dl25_25_phaseDiagram}
\end{figure}

We observe three main regions of the phase diagram, each with a given optimal form. The three configurations (shown in Fig.~\ref{fig:openChain_dl25_25_phaseDiagram}) consist of:
\begin{itemize}
    \item Configuration $\mathbf{A}$: broadly represents a compact geometry, a typical structure of low energy (pictured) forms a basic over and under crossing, the rest of the length of the curve winds around akin to a single helix doubling back on itself. Appears at medium to large solvent radius across all packing fractions.
    \item Configuration $\mathbf{B}$: a double helix, where the string folds back on itself. Occurs for smaller solvent radius, and is more typical for higher packing fractions.
    \item Configuration $\mathbf{C}$: overhand knot. Appears where the solvent radius is small across all packing fractions. 
\end{itemize}

The phase diagram contains previously unseen optimal curve shapes challenging the established idea that the $\alpha$-helix and $\beta$-sheet are the most energetically favourable shapes among biopolymers from the persepctive of solvation free energy~\cite{SNI07,HAN07}. 
The appearance of the overhand knot as a stable solvation free energy minimiser could provide a basis for the existence of knotted configurations in biopolymers.

To explore the relative energies of these structures consider two cross-sections through the phase diagram, at $r_s=0.04$ and $r_s=0.125$, shown in Fig.~\ref{fig:energy_plots_all_shapes}.
The energies are plotted against the fluid packing fraction $\eta$ for example shapes of the groups $\mathbf{A}$, $\mathbf{B}$ and $\mathbf{C}$, including two additional configurations, the open tight helical curve of lowest energy within the collection of all tight helical curves, and a low energy compact shape without otherwise recognizable structure (arising from the computer experiments).
In both plots for low packing fractions all of the structures have comparable free energy values, whereas for larger packing fractions, the structures differentiate from each other significantly. 
This suggests that solvation forces become increasingly relevant to folding dynamics as the solvent becomes denser and the size of the solute larger. 
This places particular importance on the configurations $\mathbf{B}$ and $\mathbf{C}$, the double helix and the overhand knot shape, which have significantly lower energy than other structures in the region where they are most favorable. These plots also show that the energy values of tight (single) helices are well above those of the minimising configurations identified here. 
A full phase diagram of tight helical configurations, periodic and finite curves of the same length as the solute geometries considered here, are given in the supplementary materials. 
Single helices assemble in experiments in the region A for packing fraction $\eta \lessapprox 0.15$ with energy values close to those curves of lowest energy, in agreement with the energy data shown in Fig.~\ref{fig:energy_plots_all_shapes} (Bottom).
We note that the free energy of a $\beta$-sheet configuration for a finite string is significantly higher than these values, and thus not important for our study.
Finally for the fluid conditions considered in this work the fully solvated state is not a favourable configuration so that the solvent is, to a certain degree, shape determining.

\begin{figure}[h!]
\centering
\includegraphics[width=0.8\linewidth]{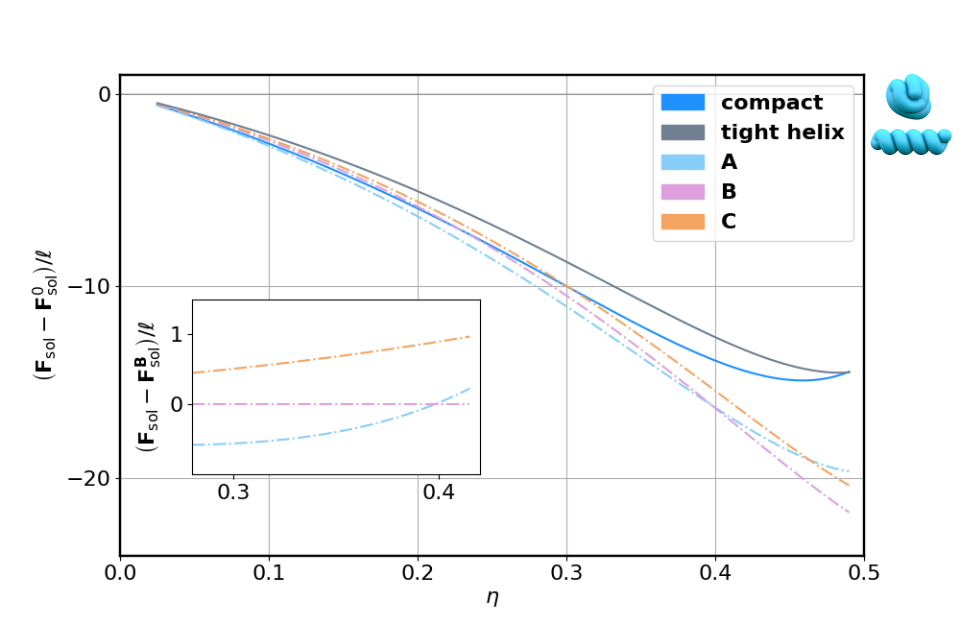}
\includegraphics[width=0.8\linewidth]{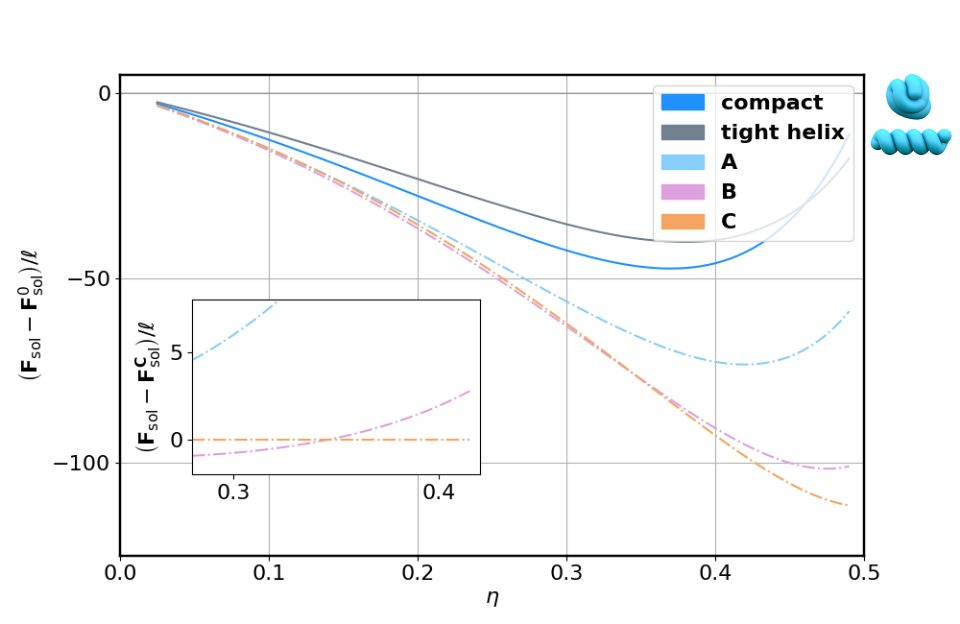}
\caption{Free energy of solvation profiles for fixed solvent radius and increasing solvent packing fraction.
Plotted is the free energy~\eqref{eq:free_energy_solvation} against the packing fraction $\eta$ for solvent radius $r_{s} = 0.125$ (Top) and $r_{s} = 0.04$ (Bottom) in units $(\mathbf{F}_{\mathrm{sol}}-\mathbf{F}_{\mathrm{sol}}^{0})/\ell$ (energy per unit length relative to the \emph{fully solvated} state).
Configurations $\mathbf{A}$, $\mathbf{B}$, $\mathbf{C}$, the tight helix of lowest energy (including the optimal helix pictured), and a generic compact shape (pictured) are compared.
Energy differences between all structures for low $\eta$ are marginal.
(Top) The dominant curve shape for low $\eta$ is $\mathbf{A}$ then $\mathbf{B}$ (double helix) as $\eta$ increases.
Inset depicts the energy of $\mathbf{A}$, $\mathbf{B}$ and $\mathbf{C}$ normalised with respect $\mathbf{B}$.
(Bottom) The dominant curve shape is $\mathbf{B}$ (double helix) for low $\eta$ and $\mathbf{C}$ (overhand knot) for $\eta \gtrapprox 0.35$. 
As $\eta \gtrapprox 0.2$ increases the shapes $\mathbf{B}$ and $\mathbf{C}$ are of significantly lower energy as compared to the other configurations.
Inset depicts the energy of $\mathbf{A}$, $\mathbf{B}$ and $\mathbf{C}$ normalised with respect to $\mathbf{C}$.}
\label{fig:energy_plots_all_shapes}
\end{figure}

\subsection{Discussion}
In this work we utilise the morphometric approach to investigate thermodynamically favourable configurations of a tubular string solute in a hard sphere solvent, as a test case of a short biopolymer living in a fluid environment of the cell under physiological conditions.
Our computer experiments demonstrate the folding of the solute into a variety of helical forms using only the interaction with the solvent. 
We find that the overhand knot and double helix configurations are the most stable structures in the large region of the phase diagram where solvation plays the most significant role, proving more stable than the optimal helix seen in protein $\alpha$-helices.

For small solutes in low density fluids (Fig.~\ref{fig:energy_plots_all_shapes} (top)) there is little energy difference between generic compact curve configurations and tight helical geometries like the open $\alpha$--helix configuration. 
This implies that hard sphere fluids of large solvent radius and low packing fraction make minimal interference with respect to the solute configuration.
In biomaterials like proteins, the energetic stability of the $\alpha$--helix comes from both the solvation free energy and hydrogen bonding of the backbone chain of protein molecule and interactions of the side chain atoms~\cite{KAL98}.
Our experiments show that for short strings, single helices are simply unfavourable in comparison to other geometries; this indicates that for short linear-chain molecules energetic gains of origin other than the solvent are the significant factor in determining the formation and prevalence of the $\alpha$-helix.
The fine balance of energies determining the native state of short linear--chain molecules should be different for structures similar to the double helical configuration $\mathbf{B}$ and overhand knot $\mathbf{C}$ and our study of the solubility of these structures provides a key insight into the properties of these geometries.  
In particular, the overhand knot appears as an optimal configuration for energies other than solvation, namely for stiffness regimes of semi--flexible polymers~\cite{MAR16, MAJ21},  and for the balance between bending energy and polymer entropy, where here the tube is viewed in the context of polymer tube theory as a tube containing entangled polymers~\cite{DAI20}.

The empirical classification of the optimal curve shapes becomes increasingly coherent with decreasing solvent radius and increasing packing fraction. 
Within this part of the diagram two optimal shapes are distinguished, a double helix and an overhand knot. 
There is no apparent shape variation among the overhand knot configurations observed in region $\mathbf{C}$. 
This shape, which is of significantly less energy than other configurations within region $\mathbf{C}$ has, to the best of our knowledge, not yet been recognised for its thermodynamic stability. 
The past years have seen extensive interest in the science of knotted proteins and biopolymers; both from the perspectives of design and function~\cite{ASH22, TEZ22, SUL20}. 
The observation that a short homogeneous tube will adopt a knotted structure based only on interactions with a solvent is an important step in understanding folding and self--assembly in aqueous environments. 
Because of the propensity for the tube to come into self--contact, these structures also relate to ideal knots, which are length minimising configurations of knotted loops. 
In our work we use the reach to constrain the curve trajectory such that the tube is at most touching but not overlapping, which is similar to computational simulation of ideal knots \cite{KAT96}.
In essence, the minimisation of length (or maximisation of radius) is a related idea that utilises geometric quantities as a measure of energy \cite{KAT96,STA96,GON99}, with a deep correlation to DNA electrophoresis \cite{STA96}, and consequences for mechanics \cite{JOH2021,GRA21}.

The double helical motif, an ubiquitous shape in biology, is reminiscent here of the plectoneme geometry adopted by supercoiled DNA~\cite{BRO18}. 
As with single helices these shapes are interesting from many different energetic perspectives~\cite{KAS10}. 
Our attention to solvation in the simple setting of our experiments confirms the thermodynamic stability of double helical geometry, in particular over other configurations.

The major challenge of this work is the computation time needed for experiments, with this hindering a similarly in depth study of longer solute geometries. 
The computation time is effected by both the solvent radius and the packing fraction with increased duration for low solvent radius and large packing fraction, in turn this effects the reliability of the results in the corresponding regions of the phase diagram (Fig.~\ref{fig:openChain_dl25_25_phaseDiagram}).
For $r_{s} < 0.04$ over all packing fractions experiments were initialised in configurations similar to the early experiment structures shown in Fig.~\ref{fig:deformation_sequence}.
Interestingly, beyond energy specific influences on the computation time,  some curve shapes seem to fold quicker than others, in particular the overhand knot shape $\mathbf{C}$ is readily reproduced over a much larger area of the phase diagram (Fig.~\ref{fig:openChain_dl25_25_phaseDiagram}) than indicated by the corresponding region. 
Thus far it is unclear if this is providing information of the energy landscape surrounding particular low energy configurations or of the folding pathways of such shapes.

An important consideration in this study is the string length. 
An investigation of the effect of length on the assembly process is important to realise, but is made difficult by the computational limitations discussed above.
Remarkably, whilst our empirical classification of energetically significant shapes is not technologically sophisticated, it was simple to implement.
It seems unlikely that an empirical classification can be made so easily for longer strings; the low energy configurations in region $\mathbf{A}$ of the diagram show quite some variety in shape even for the short strings considered in this work. 
Without a better understanding of low energy curve motifs, the analysis of results, even if they were readily available, may require a different mechanism of sorting than geometry alone.

By allowing the tubular geometry to freely form, our observations are a fundamental extension of previous studies of helical tubes in fluid evironments~\cite{HAN07}, or in crowded environments~\cite{SNI07}, demonstrating a much richer variety of thermodynamically stable structures.
In the limit case $\eta \rightarrow 0$ our results display similar motifs to the entropic packing of tubular strings~\cite{MAR00, POL08}, and by including all four measures in~\eqref{eq:free_energy_solvation} to define an accurate solvation model, our work may be seen as an essential progression of those studies.
Interestingly, our experimentally found optimal curve shapes bear a strong resemblance to the low energy helical forms adopted by a solvophobic tubular polymer in a solvent~\cite{POL08}.

Our findings should be considered in conjunction with previous investigations of energetically stable configurations of semi-flexible polymers~\cite{MAR16, WER17, MAJ21}.
These studies build upon a large body of work examining the collapse dynamics of homopolymers in a poor solvent e.g.~\cite{MON04, CHA01, POL02, MAJ17}.
The polymer, modelled as a chain of non--overlapping beads (monomers) connected by springs or inextensible sticks, folds according to a configurational energy given as an attractive and short--range repulsive potential between the monomers of the chain coupled with the bending energy of the chain. 
The flexibility or stiffness describes the propensity of the chain to bend. 
Broadly speaking, for certain flexibility regimes (depending on the specific polymer model, configuration energy and system temperature) the polymer will collapse into a globule structure or remain elongated.
Geometrically detailed studies demonstrate that for a given bending stiffness and system temperature a polymer may adopt a configuration similar to the favourable configuration $\mathbf{B}$ (double helix) and $\mathbf{C}$ (overhand knot) found in this work~\cite{WER17,MAJ21}.
Our study differentiates from these works in a fundamental way: Here the configuration energy is determined only from the free energy of solvation, i.e. from properties of the fluid environment, and \emph{not} properties of the polymer. 
Consequently we investigate favourble configurations of the \emph{same} polymer in \emph{different} fluids.
The tube model of the polymer employed in our study is devoid of configurational bias, neither flexible nor stiff. 
Nevertheless, the shape similarities between our findings and those of previous work are intriguing in what this suggests for the mechanisms governing the shape of tube–like polymers.
Our study demonstrates that the essence of flexibility or stiffness, as the preference of a given polymer for configuration $\mathbf{B}$ or $\mathbf{C}$ say, may be achieved by varying the properties of the fluid environment for the \emph{same} polymer. 
Given that many biological systems rely on the fluid environment to actuate their function, our study in effect says that a polymer can behave as flexible or stiff---it could really just depend on the fluid--- and in systems capable of shape rearrangement, our work sheds light on how this may be best achieved.

In summary, our geometry focused investigation provides examples of favourable biopolymer configurations via the optimisation of unconstrained flexible tubes, inaccessible within the framework of all--inclusive molecular dynamics simulations. 
By differentiating the role of solvation in biopolymer folding in this very general setting, our study helps illuminate the energetic background scenery in which all soluble biomolecules live.

\subsection{Methods}

We investigated short tubular solutes in water-like environments by way of computer simulation. 
The simulation tool functions as a geometric optimisation algorithm which modifies iteratively the shape of the solute body thereby decreasing the solvation free energy of the system.
Let the physical parameters, solvent radius $r_{s}$ and coefficients $\mathrm{p}$, $\sigma$, $\kappa$ and $\bar{\kappa}$, describing the thermodynamic interaction between the solute and solvent be given. 
The free energy of the system is computed using equation~\eqref{eq:free_energy_solvation} with the body 
\begin{equation*}
    \mathcal{B} = \bigcup_{i=1}^{n} \mathrm{B}_{r_{t} + r_{s}}(v_{i})
\end{equation*}
for vertices $\{v_{1}, \dots, v_{n}\}$ arranged linearly on an open equilateral polygon with edge length $e$, ball radius $r_{t}>\tfrac{e}{2}$ and solvent radius $r_{s}>0$.
Here $n = 101$, $e = 0.25$ and $r_{t} = 1.00778$ so that the scale--invariant length $\ell = \tfrac{25}{r_{t}} = 24.8065$, these numbers are set for comparison with other solute geometries of finer discretisation.
The vertex set $\{v_{1}, \dots, v_{n}\}$ is synonymous with the open equilateral polygon and we refer to both as simply the curve.
Our algorithm evaluates approximate solutions to the following problem:
What is the shape of a curve, representing a \emph{physical (re)arrangement} of the body $\mathcal{B}$, which minimise the free energy of solvation?

A physical (re)arrangement of the solute body is if the $r_{t}$--balls centered on the vertices of the curve curve overlap only \emph{along} the curve; such that the solute appears as a self--avoiding tube.
We characterise this mathematically using the following property~\cite{NAI92}; a curve satisfies the \emph{simple--tube property for the ball radius $r$} if for any pair $(i, j)$, $1 \leq j \leq n$ 
\begin{equation}\label{cond:simple_tube_property}
    \mathrm{B}_{r}(v_{i}) \cap \mathrm{B}_{r}(v_{j}) \subset \mathrm{B}_{r}(v_{k}) \,\, k = i+1, \dots, j -1 \, .
\end{equation}
If a curve satisfies the simple tube property for a ball radius $r$ then the four geometric measures; volume $\mathrm{V}$, surface area $\mathrm{A}$ etc. included in equation~\eqref{eq:free_energy_solvation}, are constant, depending only on the geometric parameters $n$, $e$ and $r$, \emph{independent} of the specific shape of the curve.
This is an elementary computation following from property~\eqref{cond:simple_tube_property} (see supplementary materials).
In particular, the volume of a solute body modelled by a curve satisfying property~\eqref{cond:simple_tube_property} for the radius $r_{t}$ will be constant and independent of the actual shape of the solute. 
As we are interested in energetically driven shape (re)arrangement of the solute body, we say a curve models a physical (re)arrangement of the (same) solute if the curve (fixed by the parameters $n$, $e$) satisfies the simple--tube property for the radius $r_{t}$ and restrict our attention to minimising the energy on this set.
Note that the simple--tube property models the thickness of the tube as defined by the global radius of curvature function or the reach were such functions restricted to the vertex set of the polygonal curve~\cite{GON99, CAN02, RAW03}.

The simple--tube property distinguishes an important class of solute configurations: a configuration is called \emph{fully solvated} if the curve satisfies the simple--tube property for the radius $(r_{t} + r_{s})$.
Since in this case, all four geometric measures terms of equation~\eqref{eq:free_energy_solvation} are constant and independent of the specific curve shape, the free energy of any fully solvated configuration is constant.
This energy, denoted as $\mathbf{F}_{\mathrm{sol}}^{0}$, is the zero from which we measure geometry with respect to energy, and is subtracted from the computed energy values as shown in the graphs in Fig.~\ref{fig:deformation_sequence} and in Fig.~\ref{fig:energy_plots_all_shapes}.
The simple--tube property is a discrete analog of the self--contact condition used to derive the collection of tight helical curves used as a basic test geometry in the investigation of favourable solute geometry with respect to solvation~\cite{PRZ01, GON99}.
Importantly the property establishes that our interest is in curves satisfying the simple--tube property for the radius $r_{t}$ but not for the radius $(r_{t} + r_{s})$.

The algorithm is based on the technique of parallel simulated annealing~\cite{LOU16}.
Once the energetic and geometric parameters as described above are defined, the algorithm is initialised with a chosen curve shape.
A single iteration generates a random curve close in shape to the current curve and decides whether to accept this new curve as the input of the next iteration.
Acceptance is decided by evaluating the energy difference between the current and new curves ($\Delta \mathbf{F}_{\mathrm{sol}} = \mathbf{F}_{\mathrm{sol}}^{\mathrm{new}} - \mathbf{F}_{\mathrm{sol}}^{\mathrm{cur}}$) and applying the metropolis criterium: If the new curve is of less energy it is accepted, otherwise it is accepted with probability $\mathsf{p}_{\lambda} = \exp{\left(\dfrac{-\Delta \mathbf{F}_{\mathrm{sol}}}{\lambda}\right)}$ for $\lambda \in \mathbb{R}_{>0}$.
The algorithm computes iterations of $m \approx 20 $ curves in parallel with constant $\lambda$ for a time interval typically $\approx 20$ hours.
During the simulation the parameter $\lambda$ generally decreases, this decreases the likelihood that configurations which increase the energy of the system are accepted.
Between intervals parallel systems may be duplicated or discontinued by mixing the states between the $m$ processors according to the probability
\begin{equation*}
\dfrac{\mathsf{p}_{\lambda}(j)}{\sum_{j=1}^{m} \mathsf{p}_{\lambda}(j)}
\end{equation*}
where
\begin{equation*}
\mathsf{p}_{\lambda}(j) = \exp{\left(\dfrac{\mathbf{F}_{\mathrm{sol}}^{j} - \mathbf{F}_{\mathrm{sol}}^{0}}{\lambda}\right)}
\end{equation*}
and $\mathbf{F}_{\mathrm{sol}}^{j}$ is the free energy of the curve of the $j$\textsuperscript{th} processor. 
The energy is computed using the POWERSASA software which evaluates the measures of~\eqref{eq:free_energy_solvation} exactly with state of the art efficiency~\cite{KLE11}.

The algorithm generates a random curve close in shape to the current curve via so--called crankshaft deformations~\cite{MIL03}--- two vertices are chosen randomly along the curve and the vertices between these two are rotated a random angle about the axis defined by the line connecting the two vertices. Here basic adaptations are included to ensure that the end vertices of the chain are translated as often as the interior vertices of the chain. 
This deformation preserves the number of vertices and edge lengths in a straight forward way but may produce a curve which does not model the solute body i.e. does not satisfy the simple--tube property for the radius $r_{t}$.
A second procedure checks if the curve bends too much causing kinks or different sections of the tube overlap, implying the simple--tube property for the radius $r_{t}$ is violated, in which case the curve is discarded and a new random curve is generated. 
This process is comparable to the computation of thickness of discrete knots and links when minimising for ropelength~\cite{RAW03}.

We are interested in comparing thermodynamically favourable solute geometries between different physiological fluid environments of biopolymers.
We achieve this by using explicit formulas of the coupling coefficients $\mathrm{p}$, $\sigma$, $\kappa$  and $\bar{\kappa}$ in terms of the packing fraction of the fluid $\eta$ and the relative size of the solvent $r_{s}$~\cite{HAN06b}.

These formulaic expressions are as follows;

\begin{align}
\mathrm{p} &= \left(\frac{3}{4\pi} \frac{1}{r_{s}^{3}}\right) \eta \frac{1 + \eta + \eta^{2} - \eta^{3}}{(1 - \eta)^{3}} \\ \nonumber
\sigma &= -\left(\frac{3}{4\pi} \frac{1}{r_{s}^{2}}\right)\eta \left( \frac{1 + 2\eta + 8\eta^{2} -5\eta^{3}}{3(1 - \eta)^{3}} + \frac{\ln(1 - \eta)}{3\eta} \right) \\\nonumber
\kappa &= \left(\frac{3}{4\pi} \frac{1}{r_{s}}\right) \eta \left( \frac{4 -10\eta + 20\eta^{2} - 8\eta^{3}}{3(1 - \eta)^{3}} + \frac{4\ln(1 - \eta)}{3\eta} \right) \\ \nonumber
\bar{\kappa} &= \left(\frac{3}{4\pi}\right) \eta \left( \frac{-4 + 11\eta - 13\eta^{2} + 4\eta^{3}}{3(1 - \eta)^{3}} - \frac{4\ln(1 - \eta)}{3\eta} \right) \,.
\end{align}

The code can be accessed from Github~\cite{disSolve}.

\bmhead{Acknowledgements}
Funded by the Deutsche Forschungsgemeinschaft (DFG - German Research Foundation) - Project-ID 195170736 - TRR109. We thank Roland Roth for discussions on the foundations of the morphometric approach and helix self assembly. We thank Roman Unger for assistance with the parallelisation of the code and running support with respect to the computer cluster of the mathematics department at the TU Chemnitz.

\bmhead{Author Contributions}
The research was conceived and designed by RC and ME, computational implementation by RC, discussion and interpretation by RC and ME, and written by RC and ME.

\bmhead{Competing Interests}
The authors declare that there are no competing interests.

\printbibliography

\newpage
\section*{Supplementary Materials}

\subsection*{Investigation of tight helical curves from the perspective of solvation.}

We consider periodic helical curves satisfying a close packing condition of the tube, whose center line is the curve and cross--sectional radius $r_{t}$. 
Such a tube will self--intersect if either the local radius of curvature of the curve is less than $r_{t}$ or a pair of points of the curve belonging to distinct segments are closer than $2r_{t}$.
Mathematically these conditions are given using either the reach functional or the global radius of curvature functional and are important curve properties with regard to the \emph{ropelength} problem in the context of knot theory~\cite{GON99, CAN02, LIT99}. 
Hence, \emph{assuming} self--contact of the tube implies either the local radius of curvature of the centerline curve is $r_{t}$ or at least two points belonging to distinct arcs of the curve a brought to a distance $2r_{t}$.
In assuming either or both of these conditions a collection of helical curves are derived for which the tube comes into self--contact but does not intersect~\cite{PRZ01}.
The collection essentially interpolates between the so--called \emph{optimal} $\alpha$--helix unwinding towards the  $\beta$--sheet and therefore provides an excelled test--case geometry for the investigation of helical tubes within a fluid, organised into the phase diagram Fig.~\ref{fig:perodic_helix_phase_diag}.
Each helix of the family is given by the (scale--invariant) helical radius $\tfrac{R_{h}}{r_{t}}>0$.
If $\frac{R_{h}^{\ast}}{r_{t}} < 0.8689$ the curves are constrained only by the local curvature and are uninteresting with regard to solvation.
The optimal helix (of helical radius $\frac{R_{h}^{\ast}}{r_{t}} = 0.8689$ and coloured red in Fig.~\ref{fig:perodic_helix_phase_diag}) is the special case for which both self--contact conditions are met i.e. the local radius of curvature equals $r_{t}$ and consecutive turns of the helical tube rest on top of one another such that at a pair of points is brought to the distance $2r_{t}$.
If $\frac{R_{h}^{\ast}}{r_{t}} > 0.8689$ consecutive turns of the helical tube rest on top of one another but the helix is free to unwind, the limit $R_{h} \rightarrow \infty$ is an infinite stack of parallel aligned curves separated by $2r_{t}$ representing the $\beta$\emph{--sheet} (blue in Fig.~\ref{fig:perodic_helix_phase_diag}).

\begin{figure}[h!]
\centering
\includegraphics[width=\textwidth]{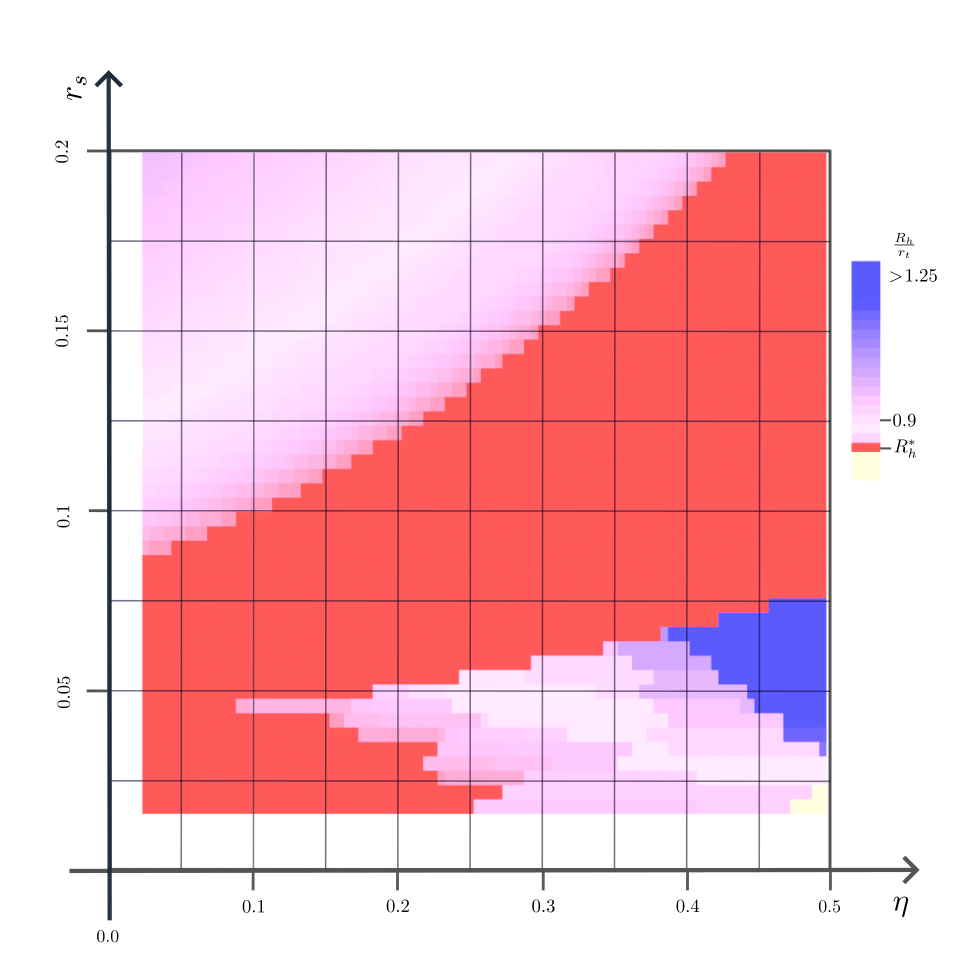}
\caption{Phase diagram of periodic tight helices with edge length $e = 0.25$. 
The diagram is a reproduction of figure 3 of \cite{HAN07} which was computed numerically over a smooth curve.
The diagram shows three main regions, a slightly unwound helix is proferred for larger solvent radii. 
A thick diagonal stripe of fluids favour the optimal helix.
For low solvent radius the helical curve of lowest energy is determined by the packing fraction, with larger packed fractions favouring parallel curves representing the $\beta$\emph{--sheet}.
}
\label{fig:perodic_helix_phase_diag}
\end{figure} 

Shown in Fig.~\ref{fig:perodic_helix_phase_diag} is a phase--diagram demonstrating the thermodynamic favourable closely packed periodic helical tubes.
The helical tube is thought of as immersed in a hard--sphere fluid with packing fraction $\eta$ and relative solvent size $r_{s}$.
The helical radius determining the system of least free energy, as compared to all closely packed helices, colours the corresponding square in the phase diagram.
Each closely packed helical curve is interpolated using an equilateral polygon of edge length $e=0.25$, in comparison to the open equilateral polygonal curves otherwise considered in this work.
Fig.~\ref{fig:perodic_helix_phase_diag} is a reproduction of figure~$3$ of~\cite{HAN07}, differing only in that the original diagram is constructed using smooth curves to define the solute geometry.
Both diagrams are in excellent agreement and show three main regions; a thick band of fluid environments for which the optimal helix is thermodynamically most favourable, for larger solvent radii a slightly unwound helix is preferred, for smaller solvent radii the favourability is determined by the packing fraction $\eta$ with higher density fluids favouring the $\beta$--sheet packing. 
Note the bottom right hand corner is coloured yellow referring to the fully solvated state-- a straight tube devoid of any economic packing-- this is is not seen in the original diagram and is likely a discretisation artifact.

Since we are interested in short open tubes in this work, we compute the same phase diagram for a tight helical open curves of length $\ell = 25$ shown in Fig.~\ref{fig:open_helix_phase_diag_L25}.
The diagram follows the same pattern as the diagram of the periodic helices with the fundamental difference that the $\beta$--sheet structure is not optimal for finite strings.
 
\begin{figure}[h!]
\centering
\includegraphics[width=\textwidth]{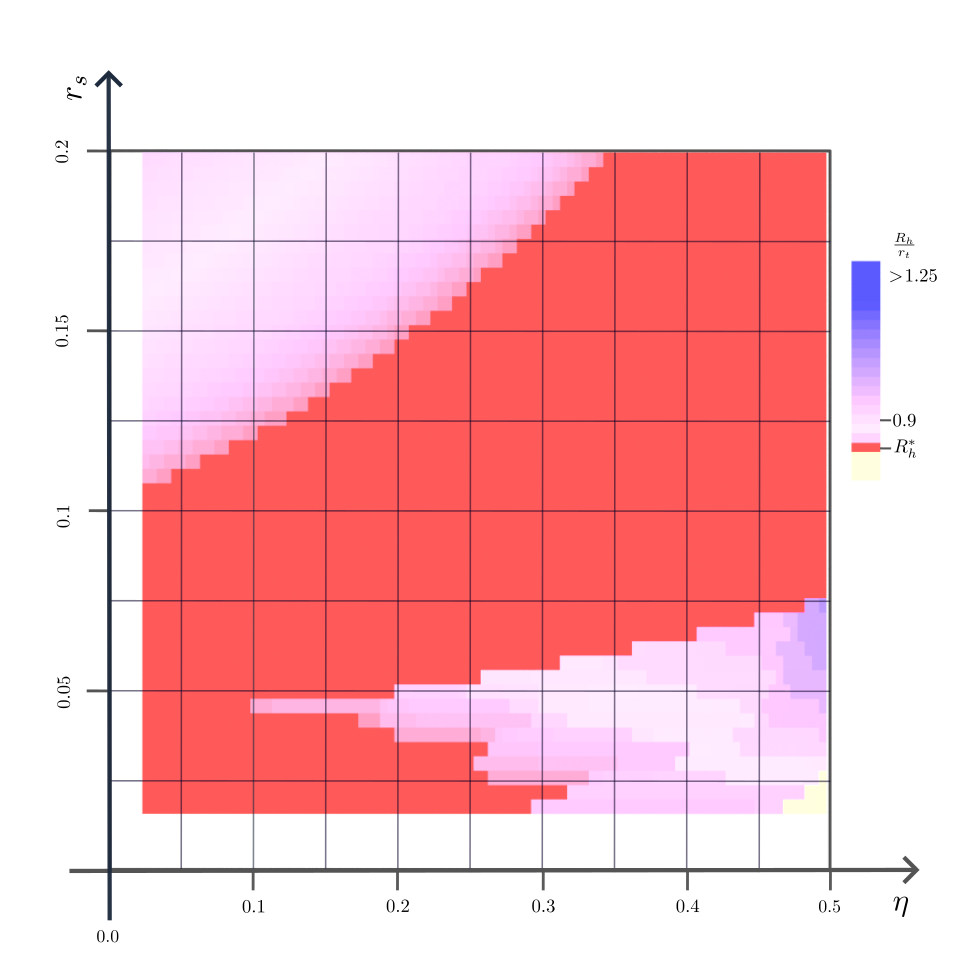}
\caption{Phase diagram of open tight helices with edge length $e = 0.25$ and total length $\ell = 25$. The diagram is different to the diagram of the periodic curves (Fig.~\ref{fig:perodic_helix_phase_diag}) as the $\beta$--sheet structure is not included.
Otherwise the trend between the various regions is quite similar.
}
\label{fig:open_helix_phase_diag_L25}
\end{figure}

\subsection*{The measures of a simple tubular string.}

We consider curves given as open equilateral polygons of $n$ vertices and edge length $e$.
The curve models the shape of the solute body as the union of $r_{t}$--balls, correspondingly the shape of body bounded by the solvent accessible surface ($\mathcal{B}$) as the union of $(r_{t} + r_{s})$--balls, centered at the vertices $v_{i}$ of the polygonal curve.
For a given fluid system the measures (volume, surface area etc.) of the body $\mathcal{B}$ determine the free energy of solvation by~\eqref{eq:free_energy_solvation}.
The solute is fully solvated  if the curve satisfies the simple tube property (condition~\eqref{cond:simple_tube_property}) for the radius $r=(r_{t} + r_{s})$.
This is a geometric condition and does not determine a unique curve shape, referring instead to curves devoid of any economic packing.
The energy of the fully solvated state, $\mathbf{F}_{\mathrm{sol}}^{0}$, is a constant depending \emph{only} upon the geometric parameters of the tubular solute and physical coefficients defining the energy in~\eqref{eq:free_energy_solvation}.
Very generally, if a curve satisfies the simple tube property for the radius $r$ then the measures of the tubular string are constant depending only on the geometric parameters $n$, $e$ and $r$ and not on the specific shape of the curve.
This observation is given in the following proposition.
By this one can see that a curve models the solute if the curve satisfies the simple tube property for the radius $r=r_{t}$ since then the volume, i.e. the bulk amount of solute thought as added to the fluid, remains fixed throughout the simulation.\\

\begin{prop}[measures of a simple tube]\label{prop:embeddedMeasuresOfThreadedBead}
The measures of an open equilateral polygonal curve with edge length $e$ and vertex set $\{v_{1}, \dots, v_{n}\}$ satisfying condition~\eqref{cond:simple_tube_property} for ball radius $r$ are given by the following formulas: 
	\begin{equation*}
	\begin{array}{rcl}
		\mathrm{V}\left( \bigcup\limits_{i=1}^{n} \mathrm{B}_{r}(v_{i})\right) &= &\pi (n-1)\left(r^{2}e - \frac{e^{3}}{12}\right) + \frac{4 \pi}{3} r^{3}\\[16pt]
	  	\mathrm{A}\left( \bigcup\limits_{i=1}^{n} \mathrm{B}_{r}(v_{i})\right) &= &2(n - 1)\pi r e + 4\pi r^{2}\\[16pt]
		\mathrm{C}\left( \bigcup\limits_{i=1}^{n} \mathrm{B}_{r}(v_{i})\right) &= &2(n - 1)\pi\left(e - \sqrt{r^{2} - \left(\frac{e}{2}\right)^{2}} \left(\frac{\pi}{2} - \arccos\left(\frac{e}{2r}\right)\right) \right) + 4\pi r \\[16pt]
		\mathrm{X}\left( \bigcup\limits_{i=1}^{n} \mathrm{B}_{r}(v_{i})\right) &= & 4\pi \,.
	\end{array}
	\end{equation*}
\end{prop}
\begin{proof}
    The Voronoi diagram intersected with the ball $\mathrm{B}_{r}(v_{i})$ for each $i$ is a decomposition of the union of balls which, by~\eqref{cond:simple_tube_property}, is disjoint up to a planar face interior to the union of balls. 
    A set belonging to such a decomposition is coloured blue in Fig~\ref{fig:linear_chain_solute}.
    The volume of this region (blue) corresponding to the string interior vertices $v_{2}, \dots, v_{n-1}$ is the volume of a ball without the volume enclosed by two spherical caps and the bounding planes (shaded grey); that is
    \begin{equation*}
      \frac{4\pi}{3} - \frac{2\pi}{3}\left(r - \tfrac{e}{2}\right)^{2}\left(2r - \tfrac{e}{2}\right) = \pi(r^{2} - \frac{e^{3}}{12}) \,.
    \end{equation*}
    The tubular string is the concatenation of $(n - 2)$ such regions (blue).
    The volume of the regions corresponding to the $i=1$ and $i=n$ vertices is the volume of the ball without the volume enclosed by one spherical cap and a bounding plane; so that the addition of both of these volumes is 
    \begin{equation*}
      2\frac{4\pi}{3} - \frac{2\pi}{3}\left(r - \tfrac{e}{2}\right)^{2}\left(2r - \tfrac{e}{2}\right) = \frac{4\pi}{3} + \pi(r^{2} - \frac{e^{3}}{12}) \,.
    \end{equation*}
    In sum this gives the computed volume as in the proposition.

    Completely analogously the (exposed) surface area of the region (blue) corresponding to the string interior vertices $v_{2}, \dots, v_{n-1}$ is the area of the ball without the area of two spherical caps; that is
     \begin{equation*}
      4\pi r^{2} - 4\pi r\left(r - \tfrac{e}{2} \right) = 2\pi r e \,.
    \end{equation*}
    The total surface area of the tubular string is the area of $(n - 2)$ such regions and the surface area of both regions corresponding the $i=1$ and $i=n$ vertices, which sum as 
    \begin{equation*}
      8\pi r^{2} - 4\pi r\left(r - \tfrac{e}{2} \right) = 4\pi r^{2} + 2\pi r e \,.
    \end{equation*}
    In total this gives the computed surface area as in the proposition.

    The integrated mean curvature measure has a term which is computed over the regular surface of the the union of balls and a term computed over the intersection curves between adjacent balls.
    Since both principal curvatures are equal to $\frac{1}{r}$ the contribution to the integrated mean curvature from integration over the regular surface is
    \begin{equation*}
        4\pi r + (n - 1)2\pi e \,. 
    \end{equation*}
    The contribution to the integrated mean curvature measure computed over the intersection curves between adjacent balls is (minus) half the angle between the normal vectors to intersecting surfaces, $\theta$, integrated along the intersection curve $\gamma$; that is
    \begin{equation*}
        - \oint\limits_{\gamma} \frac{\theta}{2} \,.
    \end{equation*}
    By symmetry the angle $\theta$ is constant along the intersection curve $\gamma$ and equal to
    \begin{equation}
        \pi - 2\arccos\left( \frac{e}{2r} \right) \,.
    \end{equation}
    The length of the intersection curve $\gamma$ between two arbitrary adjacent balls is
    \begin{equation*}
        2\pi\sqrt{r^{2} - (\tfrac{e}{2})^{2}} \,.
    \end{equation*}
    There are $(n-1)$ such intersection curves contained in the boundary of the tubular string.
    In total these contributions give the computed mean curvature term as in the proposition.

    The integrated Gaussian curvature is computed via the Euler characteristic and observing that condition~\eqref{cond:simple_tube_property} ensures that the tubular string has the topology of a ball.
\end{proof}

\begin{figure}[h!]
\centering
\includegraphics[width=0.9\textwidth]{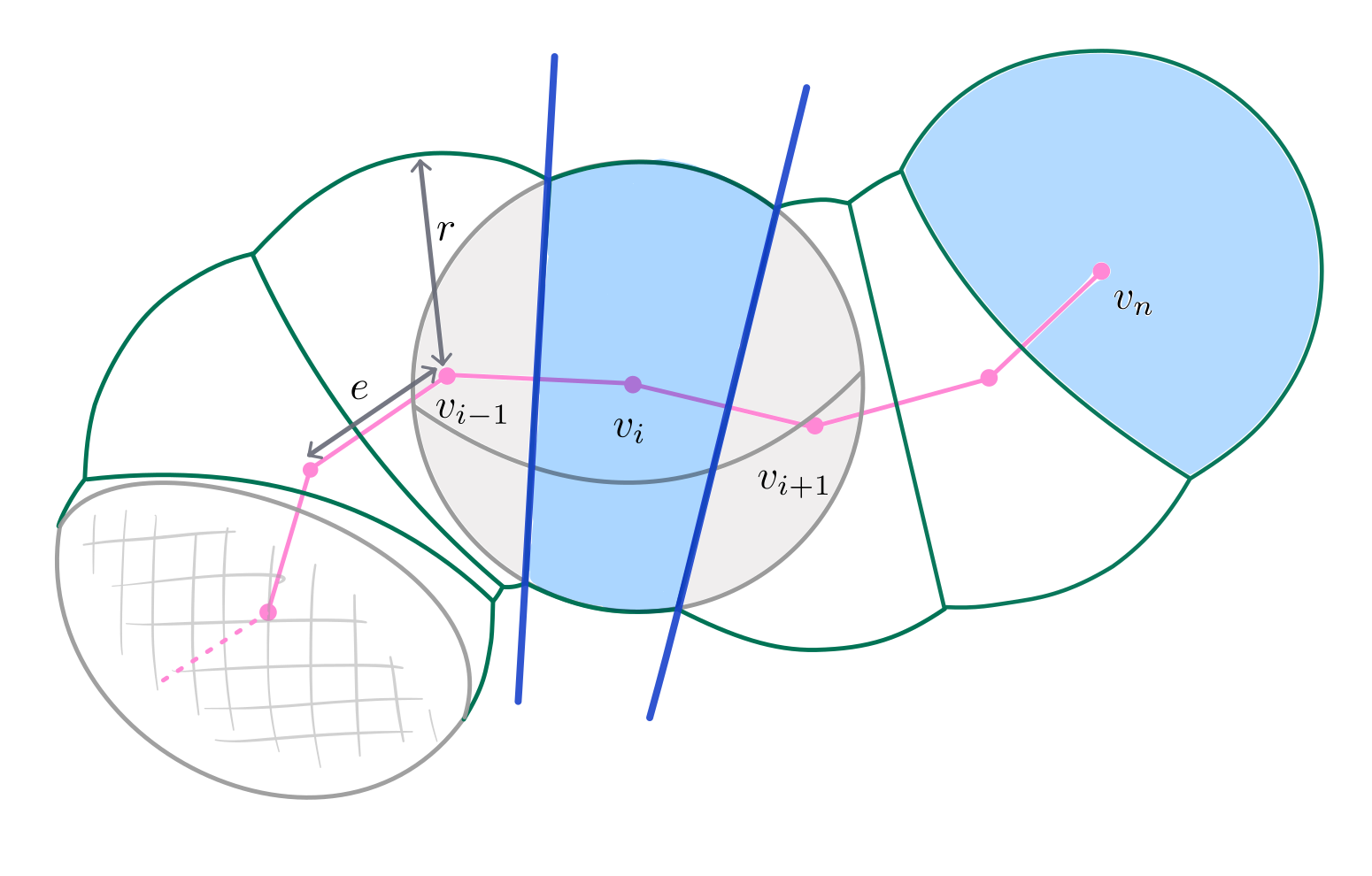}
\caption{The figure depicts part of the body $\mathcal{B}$, here given as the union of $r$ balls centered along the vertices of an equilateral polygonal curve with $n$ vertices and edge length $e$. 
In this configuration, the curve satisfies the simple tube property for the radius $r$, condition~\eqref{cond:simple_tube_property}.
To compute the measures of the union of balls a disjoint decomposition of $\mathcal{B}$ is used.
Such a disjoint decomposition is found by intersecting the $i$\textsuperscript{th} region of the the Voronoi diagram of the vertex set $\{v_{1}, \dots, v_{n}\}$ with the ball $\mathrm{B}_{r}(v_{i})$.
Shaded blue is an example of such a set for the interior vertices $\{v_{2}, \dots, v_{n-1}\}$ and an example of such a set for end vertices $v_{1}, v_{n}$.
Each such set is the ball $\mathrm{B}_{r}(v_{i})$ contained in the half space defined by the perpendicular bisecting planes (depicted as dark blue lines) of the incident edges to the vertex.
}
\label{fig:linear_chain_solute}
\end{figure}

\end{document}